\documentclass[journal]{IEEEtran}
\usepackage{amsfonts}
\usepackage{algorithmic}
\usepackage{algorithm}
\usepackage{array}
\usepackage[caption=false,font=normalsize,labelfont=sf,textfont=sf]{subfig}
\usepackage{textcomp}
\usepackage{stfloats}
\usepackage{url}
\usepackage{verbatim}
\usepackage{graphicx}
\usepackage{cite}
\usepackage{booktabs}
\usepackage{xcolor,colortbl}
\usepackage{multicol}
\usepackage{pbox}
\usepackage{color,soul}
\usepackage{amsmath, bm}
\usepackage{amssymb}
\usepackage{amsthm}
\usepackage{mathrsfs}
\usepackage{dutchcal}
\usepackage{array}
\usepackage{graphicx, nicefrac}
\newtheorem{theorem}{Theorem}
\newtheorem{lemma}{Lemma}
\newtheorem{proposition}{Proposition}
\begin{document}

\title{RIS-Assisted D2D Communication in the Presence~of Interference: Outage Performance Analysis and DNN-Based Prediction}

\author{Hamid Amiriara, Farid Ashtiani,~\IEEEmembership{Senior Member,~IEEE,} Mahtab Mirmohseni,~\IEEEmembership{Senior Member,~IEEE,}\\ Masoumeh Nasiri-Kenari,~\IEEEmembership{Senior Member,~IEEE,} and Behrouz Maham,~\IEEEmembership{Senior Member,~IEEE}

\thanks{H. Amiriara, F. Ashtiani, M. Mirmohseni, and M. Nasiri-Kenari are with the Department of Electrical Engineering, Sharif University of Technology, Tehran 1458889694, Iran (email: {hamid.amiriara, ashtianimt, mirmohseni, mnasiri}@sharif.edu).}
\thanks{B. Maham is with the Department of Electrical and Computer Engineering, Nazarbayev University, Astana 010000, Kazakhstan (email: behrouz.maham@nu.edu.kz).}
\thanks{This work was supported by the Ministry of Science and Higher Education of the Republic of Kazakhstan, under project No. AP13068587, and by Iran National Science Foundation (INSF) under project No. 4001804.}}


\maketitle

\begin{abstract}
This paper analyses the performance of reconfigurable intelligent surface (RIS)-assisted device-to-device (D2D) communication systems, focusing on addressing co-channel interference, a prevalent issue due to the frequency reuse of sidelink in the underlay in-band D2D communications. In contrast to previous studies that either neglect interference or consider it only at the user, our research investigates a performance analysis in terms of outage probability (OP) for RIS-assisted D2D communication systems considering the presence of interference at both the user and the RIS. More specifically, we introduce a novel integral-form expression for an exact analysis of OP. Additionally, we present a new accurate approximation expression for OP, using the gamma distributions to approximate the fading of both desired and interference links, thereby yielding a closed-form expression. Nevertheless, both derived expressions, i.e., the exact integral-form and the approximate closed-form, contain special functions, such as Meijer's G-function and the parabolic cylinder function, which complicate real-time OP analysis. To circumvent this, we employ a deep neural network (DNN) for real-time OP prediction, trained with data generated by the exact expression. Moreover, we present a tight upper bound that quantifies the impact of interference on achievable diversity order and coding gain. We validate the derived expressions through Monte Carlo simulations. Our analysis reveals that while interference does not affect the system's diversity order, it significantly degrades the performance by reducing the coding gain. The results further demonstrate that increasing the number of RIS's reflecting elements is an effective strategy to mitigate the adverse effects of the interference on the system performance.
\end{abstract}


\section{Introduction}
\IEEEPARstart{T}{he} reconfigurable intelligent surface (RIS) is emerging as a promising technology in the sixth generation (6G) of wireless communications, aiming to deliver higher data rates, enhanced network coverage, and energy-efficient communication.
Characterized by their low latency and cost-effective deployment, RISs consist of electronically adjustable elements that enable modulation of the radio environment through phase shifter adjustments for enhanced signal propagation. These surfaces can seamlessly integrate into existing infrastructure, such as walls, ceilings, and buildings, to augment wireless network performance, making them a prevalent component in modern network design~\cite{ref1}.

Amidst these advancements, the anticipated rapid increase in connected devices, projected to exceed 207 billion globally by the end of 2024 \cite{ref2}, poses new challenges. This unprecedented growth presents substantial challenges in efficiently utilizing the limited spectrum and power resources amid such massive access.
Direct device-to-device (D2D) communication, a vital component of fifth-generation (5G) networks, presents a viable solution to these challenges. With the increasing need for D2D communications, 3GPP began studying and specifying D2D in Release 12, further evolving it into vehicle-to-anything (V2X) using the so-called \textit{sidelink} \cite{ref_underlay}. This sidelink is a dedicated communication link that allows direct transmission between devices without routing through the network's BS. A key design consideration for D2D and V2X communication, particularly in their underlay in-band modes, is the ability of the sidelink to reuse the same time-frequency band as the cellular uplink/downlink.

The above discussion highlights the importance of addressing co-channel interference in D2D and V2X communication, where frequency reuse is common. While the integration of RIS is designed to enhance signal coverage and quality by redirecting signals to shadowed areas, it simultaneously creates additional paths for interference. Consequently, the challenge of co-channel interference intensifies with the incorporation of RIS. This scenario underlines the critical need to consider not only the direct impact of co-channel interference on users but also the reflected interference stemming from the RIS.

\subsection{Related Works and Motivation}
Several recent studies have investigated the impact of interference on the performance of RIS-assisted communications. However, these studies concentrate solely on the user's experience of interference, neglecting the effects of interference at the RIS \cite{ref3,ref4,ref5,ref6,ref7,ref8,ref9,ref10}.

In \cite{ref3}, the application of RIS in wireless-powered interference-limited communication systems was explored, revealing a high dependence of achievable performance on interference level. In \cite{ref4} and \cite{ref5}, the authors examined an RIS-assisted cellular network, where the destination is subject to co-channel interference. They analyzed the impact of co-channel interference on the system performance in a Rayleigh fading channel.
The study in \cite{ref6} delved into the performance of an RIS-aided mixed FSO/RF system in the presence of co-channel interference, particularly for a large number of reflecting elements, using the non-central chi-square distribution to model the cascaded channels.
Furthermore, \cite{ref7} presented the performance analysis of an RIS-aided relaying system under interference, utilizing the RIS as a transmitter.
The authors in \cite{ref9} analyzed the ergodic capacity and energy efficiency for the downlink signal received by mobile users in the RIS-aided cellular system, considering interference from co-channel base stations (BSs) in cell-edge users.
In \cite{ref8}, the authors derived an analytical upper bound for the ergodic achievable rate of a D2D communication system assisted by RIS over a fading channel, considering interference between cellular users and the D2D receiver.
In \cite{ref10}, the author proposed a deep reinforcement learning (DRL) algorithm to address the joint optimization of power allocation and phase shift matrix. This approach aims to mitigate interference from other D2D transmitters at the target D2D receiver in RIS-assisted D2D communications.
In \cite{New1}, the performance of an RIS-aided mobile destination node in the presence of mobile co-channel interferers was evaluated under Rayleigh fading conditions. Building on \cite{New1}, the author in \cite{New2} extended the analysis by incorporating $\kappa-\mu$ fading distributions and transceiver hardware distortions.
In \cite{New3}, the author derived the outage probability of an RIS-assisted communication system, investigating scenarios with co-channel interference sources around the user.

However, these studies do not fully address the impact of co-channel interference on system performance in the context of RIS-aided D2D communication. Since the RIS cannot suppress interference\footnote{Owing to the lack of precise knowledge about the interference channel in practical scenarios.}, it becomes crucial to analyze the influence of reflected interference on system performance. Consequently, it is essential to consider the impact of interference both at the RIS and at the user.

Outage probability (OP) stands as a key metric for measuring communication performance, and accurate OP prediction is crucial for enhancing reliability. Nevertheless, evaluating the OP performance in RIS-assisted communications proves challenging, particularly when considering interference at both the RIS and the user.
To tackle this issue, prior works on interference-free RIS-assisted communications have often employed various approximation methods to model the cascaded channels; such as the central limit theorem (CLT)-based approximation \cite{ref11,ref12,ref13} and gamma-based approximation \cite{ref14,ref15,ref16}. The non-central Chi-squared distribution approximation method, relying on the CLT, is predominantly accurate when there is a large number of reflecting elements involved. However, this method can result in significant approximation errors in scenarios with a high signal-to-noise ratio (SNR) regime \cite{ref17}. To circumvent these CLT issues, the authors in \cite{ref12}, \cite{ref13} used gamma distribution to approximate the fading of each reflecting path. On the one hand, the gamma distribution encompasses many power distributions as special cases, and on the other hand, it allows for tractability when evaluating the outage. The gamma-based framework appears to offer more accurate results than the CLT approximation \cite{ref34}. However, the gamma-based approximation approach provides neither a bound nor an asymptotic result for the OP \cite{ref_bound}.

Nevertheless, precisely analyzing OP performance in real-time remains a challenge in RIS-assisted D2D wireless communication networks, especially with the presence of interference at both the RIS and user ends. Real-time analysis is crucial in dynamically changing D2D environments, where rapid changes in network conditions require immediate adjustments to maintain optimal performance and ensure quality of service (QoS). An alternative approach is to predict the performance using machine learning (ML) techniques, which is beyond system optimization and design \cite{ref18}, \cite{ref19}. It has been increasingly applied with success in real-time prediction of the performance across various wireless communication networks \cite{ref20,ref21,ref22,ref23,ref24}. Xu et al. employed a back-propagation (BP) model to forecast bit error probability for mobile internet of things communication systems in \cite{ref20}. A BP model for feature fusion was utilized to identify unknown emitters in \cite{ref21}. In \cite{ref22}, the author proposed a deep neural network (DNN)-based OP intelligent prediction algorithm for the mobile communication internet of things. In~\cite{ref23}, an intelligent OP prediction approach using an Elman model was proposed to evaluate communication performance in internet of vehicle networks. Additionally, \cite{ref24} presented an intelligent OP prediction algorithm based on the improved cuckoo search \cite{ref_cuckoo} for vehicle-to-vehicle communication in mobile networks.

\subsection{Our Contributions}
Based on the aforementioned discussion, this paper presents a new mathematical analysis for the OP performance of RIS-assisted communications in more practical scenarios where both the RIS and the sidelink user are impacted by interference. To precisely represent the OP performance of the considered network, integral-form expressions have been derived. Furthermore, we provide closed-form expressions for approximating the OP performance by using the gamma distribution to model the cascaded channels, which is applicable to an arbitrary number of reflecting elements. However, the derived closed-form expressions involve the parabolic cylinder function, posing computational challenges for mobile networks and hindering the real-time analysis of OP performance.
To address this, we introduce a DNN-based approach for predicting the exact OP performance in real-time. Moreover, we establish tight upper bounds for the OP performance, facilitating the assessment of achievable diversity order and coding gain.

The main contributions of this paper are summarized as follows:
\begin{itemize}
\item{We derive an exact OP expression (in integral form) for an RIS-aided D2D wireless system, considering the interference at both the user and the RIS}
\item{We develop a closed-form expression for the system's OP utilizing a tractable gamma distribution approximation.}
\item{We provide a tight upper bound for the exact OP performance, enabling the characterization of the achievable diversity order and the coding gain in the considered system.}
\item{We propose a novel, exact OP prediction algorithm based on DNN. This proposed OP performance prediction approach can evaluate communication performance in real-time.}
\item{Finally, we evaluate the derived exact OP expression and the gamma-based OP approximation using Monte Carlo simulations. Our results reveal that the upper bound aligns closely with the simulation results, confirming the theoretical analysis. Our simulation analysis further indicates that the adverse effects of the interference can be significantly reduced by employing a large RIS. Moreover, the results show that while interference notably reduces the coding gain, it does not affect the system's diversity order. Additionally, the simulations demonstrate that the proposed DNN-based OP prediction method provides highly accurate results with the added benefit of reduced execution time.}
\end{itemize}
\subsection{Organization and Notations}
The rest of the paper is organized as follows. Section~\ref{sec:system_model} introduces the system model of RIS-assisted D2D communication systems. Section~\ref{sec:op_analysis} offers an extensive analysis of OP, encompassing the derivation of exact expressions, gamma-based approximations, and an asymptotic analysis of OP performance. Section~\ref{sec:dnn_approach} details the development and implementation of a DNN-based approach for predicting OP. In Section~\ref{sec:simulation_results}, we present our Monte Carlo simulations and numerical results, showcasing the accuracy of the derived analytical expressions. Finally, Section~\ref{sec:conclusion} concludes the paper.

\textit{Notations:} Italicized letters (e.g., $a$) indicate scalars, bold lowercase letters (e.g., $\mathbf{a}$) represent vectors, bold uppercase letters (e.g., $\mathbf{A}$) denote matrices, and uppercase script letters (e.g., $\mathcal{A}$) signify sets.
For any complex vector $\mathbf{a}$, $\angle \mathbf{a}$ refers to the phase vector, with each element corresponding to the phase of the respective element in $\mathbf{a}$, while $\mathbf{a}(n)$ refers to its $n$-th element.
The symbols $\mathbb{E}[\cdot]$, $\text{Var}[\cdot]$, $f_X (\cdot)$, and $\mathscr{M}_X(\cdot)$ denote the expectation, variance, probability density function (PDF), and moment-generating function (MGF) of a random variable (RV) $X$, respectively.
Exponential and logarithmic functions are expressed as $\exp(\cdot)$ and $\log(\cdot)$.
Furthermore, $x!$, $\gamma(a,x)$, $\Gamma(x)$, $G_{p,q}^{m,n}\left[\left.x\right|_{b_1,\ldots,b_q}^{a_1,\ldots,a_p}\right]$, $D_a(x)$, $\Phi(a,b;x)$, and $_pF_q(a_1,\ldots,a_p; b_1,\ldots,b_q; x)$ signify the factorial of $x$, the incomplete gamma function \cite[eq.~(8.350.2)]{ref25}, the gamma function \cite[eq.~(8.310)]{ref25}, Meijer's G-function\footnote{The Meijer's G-function is defined by the Mellin-Barnes integral as $G_{p,q}^{m,n}\left[x|_{b_1,\ldots,b_q}^{a_1,\ldots,a_p}\right]=\frac{1}{2\pi i}\int_{L}{\frac{\prod_{j=1}^{m}\Gamma\left(b_j-s\right)\prod_{j=1}^{n}\Gamma\left(1-a_j+s\right)}{\prod_{j=m+1}^{q}\Gamma\left(1-b_j+s\right)\prod_{j=n+1}^{p}\Gamma\left(a_j-s\right)}x^s ds}$.} \cite[eq.~(9.301)]{ref25}, parabolic cylinder functions\footnote{Power series for parabolic cylinder function have been presented in \cite{ref26}.} \cite[eq.~(9.240)]{ref25}, the degenerate hypergeometric function \cite[eq.~(9.210)]{ref25}, and the generalized hypergeometric function \cite[eq.~(9.14.1)]{ref25}, respectively.\footnote{These special functions are built-in functions in the most renowned mathematical software packages, including MAPLE, MATHEMATICA, and MATLAB.}
Additionally, $\mathcal{I}_\nu(\cdot)$ and $\mathcal{K}_\nu(\cdot)$ are the modified Bessel functions of the first and second kind of order $\nu$ \cite[eq.~(9.6.3)]{ref25}, respectively, and $\mathcal{L}^{-1}\{\cdot; \cdot\}$ symbolizes inverse Laplace transform operator.

\section{System Model} \label{sec:system_model}
This paper considers a scenario that includes multiple D2D and cellular users within an orthogonal frequency-division multiplexing (OFDM) framework. It is assumed that each subcarrier exclusively supports a single pair of D2D users through resource-reuse, alongside a cellular user.

Without loss of generality, our focus is on a subsystem where a D2D communication pair coexists with a single cellular user on the same subcarrier, leading to potential interference. As illustrated in Fig.~\ref{fig_1}, this subsystem includes a D2D pair: a single-antenna D2D source user, denoted as ``S", and a single-antenna D2D destination user, denoted as ``D". In situations where the line of sight (LoS) link between these users is obstructed, or environmental factors impede direct communication, they communicate through an RIS, denoted as ``R", which is equipped with $N$ passive reflecting elements. Due to sharing the same time-frequency resources by the D2D users and a single-antenna cellular user (interfering user), denoted as ``I", the interference signal from the cellular user is received both at the D2D destination and the RIS.

\begin{figure}[!t]
\centering
\includegraphics[width=3in]{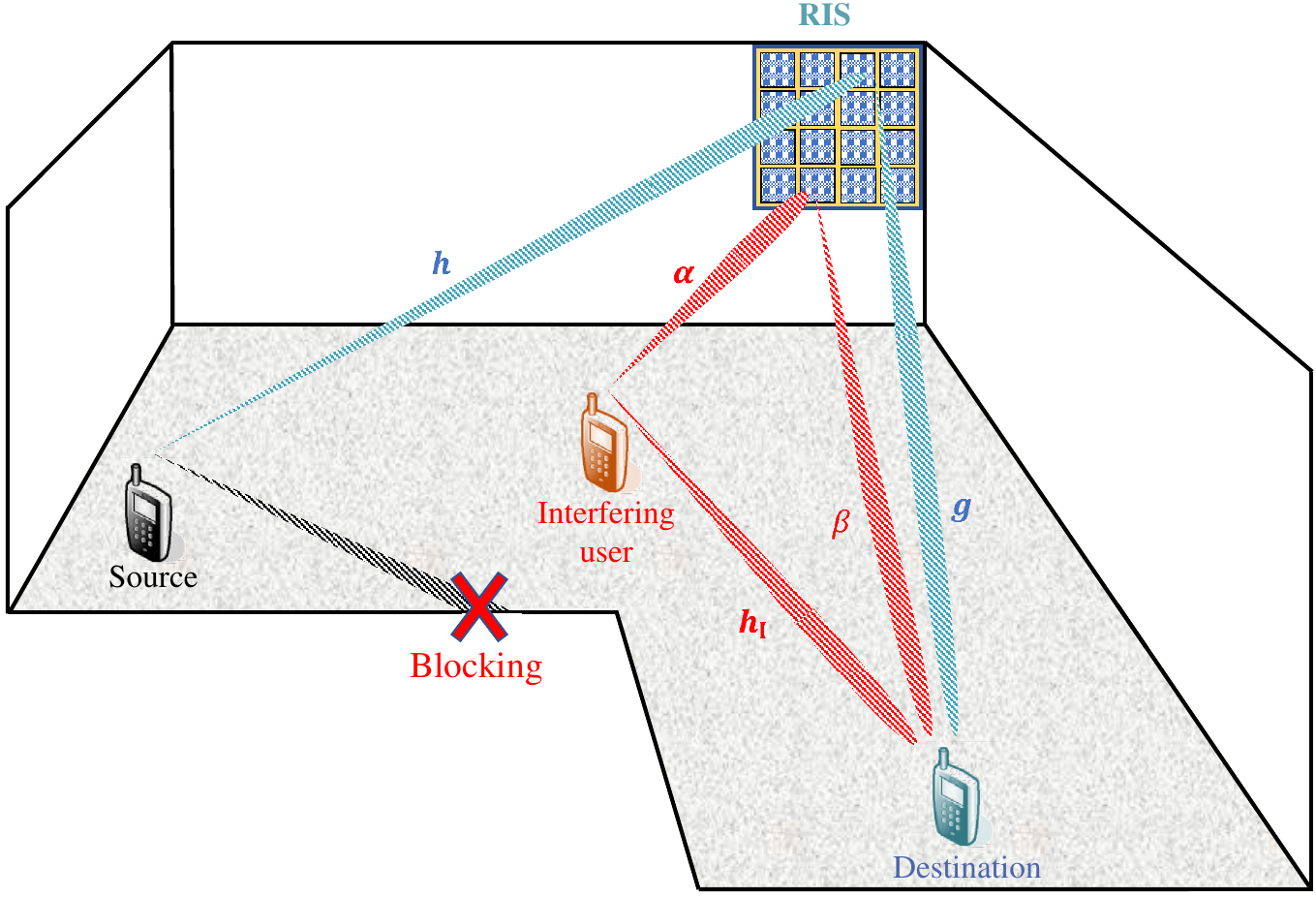}
\caption{RIS-assisted D2D communication system in the presence of interference at both the destination and the RIS.}
\label{fig_1}
\end{figure}

Consequently, the received signal at D is expressed as
\begin{equation}
\begin{aligned}
y = & \underbrace{\sqrt{p_\text{S}} \mathbf{g}^H \mathbf{\Theta} \mathbf{h} x_\text{S}}_{\text{Desired signal}}
+ \underbrace{\sqrt{p_\text{I}} \boldsymbol{\beta}^H \mathbf{\Theta} \boldsymbol{\alpha} x_\text{I}}_{\text{Reflected interference by R}}
+ \underbrace{\sqrt{p_\text{I}} {h}_\text{I} x_\text{I}}_{\text{Interference at the D}} + w,
\end{aligned}
\label{eq1}
\end{equation}
where $x_\text{S}$, $x_\text{I}$ and $p_\text{S}$, $p_\text{I}$ are the unit-norm signals and power transmitted by S and I, respectively. The first term in (\ref{eq1}) represents the desired signal, with $\mathbf{h} \in \mathbb{C}^{N \times 1}$ and $\mathbf{g} \in \mathbb{C}^{N \times 1}$, being the channel vectors from S to R and from R to D, respectively.
The second term in (\ref{eq1}) denotes the interference reflected through the RIS, where $\boldsymbol{\alpha} \in \mathbb{C}^{N \times 1}$ and $\boldsymbol{\beta} \in \mathbb{C}^{N \times 1}$ are the channel vectors from I to R and from R to D, respectively. The third term in (\ref{eq1}) represents the interference from the cellular user, i.e., I, to the D2D destination user, i.e., D, where ${h}_\text{I} \in \mathbb{C}^{1 \times 1}$ is the corresponding channel vector. Finally, $w \sim \mathcal{C}\mathscr{N}(0, N_0)$ denotes the additive white Gaussian noise (AWGN). 
In this paper, we focus on indoor D2D communication channels, predominantly characterized by non-line-of-sight (NLOS) components, using Rayleigh RVs.
Consequently, the channel $h_\text{I}$, as well as each element of the channel vectors, i.e., $\mathbf{h}(n)$, $\mathbf{g}(n)$, $\boldsymbol{\alpha}(n)$, $\boldsymbol{\beta}(n)$ $\forall n \in \mathscr{N} \triangleq \{1, 2, \ldots, N\}$, are assumed to be independent and identically distributed (i.i.d.) complex Gaussian fading with zero-mean and variances $\sigma_\text{I,D}^2$, $\sigma_\text{S,R}^2$, $\sigma_\text{R,D}^2$, $\sigma_\text{I,R}^2$, and $\sigma_\text{R,D}^2$, respectively.\footnote{It should be noted that this study could be extended to Nakagami-$m$ channel scenarios, using \cite[Theorem 1]{ref17}.}

In (\ref{eq1}), $\mathbf{\Theta} = \text{diag}\{\mathbf{r}(1)e^{j\boldsymbol{\theta}(1)}, \mathbf{r}(2)e^{j\boldsymbol{\theta}(2)}, \ldots, \mathbf{r}(N)e^{j\boldsymbol{\theta}(N)}\} \in \mathbb{C}^{N \times N}$ denotes the reflection coefficient matrix generated by the RIS, where $\mathbf{r}(n) \leq 1$ indicates the amplitude attenuation, and $\boldsymbol{\theta}(n) \in [0, 2\pi)$ represents the phase shift of the $n$-th reflecting element in the RIS. We assume that each RIS element reflects the signal without loss, i.e., $\mathbf{r}(n) = 1$, which is a common assumption due to hardware implementation and analytical simplicity. Furthermore, we assume that the RIS controller has full channel state information (CSI) of the D2D link\footnote{By using the existing channel estimation techniques \cite{ref27}, we can obtain the CSI of all channels. In the case of imperfect CSI and channel estimation errors, the proposed formulas can be used as the upper bounds of system performance.}, but there is no channel knowledge of the interference link. Consequently, the phase shift response of each element is locally optimized based on the D2D link as $\boldsymbol{\theta}^*(n) = -\angle \mathbf{g}(n) - \angle \mathbf{h}(n)$.

Thus, the instantaneous signal-to-interference plus noise ratio (SINR) at D can be written as
\begin{equation}
\gamma = \frac{\bar{\gamma}_\text{S} \left( \sum_{n=1}^{N} |\mathbf{g}(n)| |\mathbf{h}(n)| \right)^2}{\bar{\gamma}_\text{I} \left| \sum_{n=1}^{N} | \boldsymbol{\beta}(n) | | \boldsymbol{\alpha}(n) | e^{j\boldsymbol{\theta}'(n)} + h_\text{I} \right|^2 + 1},
\label{eq:sinr}
\end{equation}
where $\bar{\gamma}_\text{S} = \frac{p_\text{S}}{N_0}$ and $\bar{\gamma}_\text{I} = \frac{p_\text{I}}{N_0}$ denote the average SNR and the average interference-to-noise ratio (INR), respectively. We assume that the RIS phase shifts are specifically adjusted to enhance the D2D link, eliminating the necessity for CSI regarding the interference link. Following this approach, $\boldsymbol{\theta}'(n) = \boldsymbol{\theta}^*(n) + \angle \boldsymbol{\beta}(n) + \angle \boldsymbol{\alpha}(n)$ in (\ref{eq:sinr}) can be regarded as a uniformly distributed random variable within $[0, 2\pi)$ \cite{ref28}.

As the system suffers from interference, we assume the network operates in an interference-limited regime. This assumption implies that noise is negligible compared to the interference, allowing us to ignore the ``1" in the denominator of (\ref{eq:sinr}) \cite{ref29, ref30}. Therefore, the signal-to-interference ratio (SIR) can be given by
\begin{equation}
\gamma = \bar{\gamma}\frac{X^2}{Y^2},
\label{eq:sir}
\end{equation}
where $\bar{\gamma} \triangleq \frac{\bar{\gamma}_\text{S}}{\bar{\gamma}_\text{I}}$ denotes the average SIR, while $X$ and $Y$ are two new RVs defined as $X \triangleq \sum_{n=1}^{N} |\mathbf{g}(n)| |\mathbf{h}(n)|$, and $Y \triangleq \left| \sum_{n=1}^{N} | \boldsymbol{\beta}(n) | | \boldsymbol{\alpha}(n) | e^{j\boldsymbol{\theta}'(n)} + h_\text{I}\right|$.

\section{Outage Probability Analysis} \label{sec:op_analysis}
Assuming an SIR threshold of $\gamma_{\text{th}}$, the OP of the system is defined as the probability that the system's SIR falls below the preset threshold value. This can be mathematically written as
\begin{align}
P_{\text{out}} &= \Pr(\gamma < \gamma_{\text{th}})= \Pr\left(\frac{X^2}{Y^2} < \frac{\gamma_{\text{th}}}{\bar{\gamma}}\right)\notag\\
&= \int_{y=0}^{\infty} \int_{x=0}^{y \frac{\gamma_{\text{th}}}{\bar{\gamma}}} f_{X^2}(x) f_{Y^2}(y) \, dx \, dy,
\label{eq:pout}
\end{align}
where $f_{X^2}(x)$ and $f_{Y^2}(x)$ are the PDFs of $X^2$ and $Y^2$, respectively, and (\ref{eq:pout}) follows as $X$ and $Y$ are independent RVs.

Assessing OP requires access to the distributions of $X$ and $Y$, which are difficult to characterize. For that reason, in this section, we derive integral-form expressions for the exact OP as well as closed-form expressions for the approximate OP to analyze the performance of the RIS-assisted communication system. Additionally, an asymptotic analysis is performed, assuming a high average SIR, i.e., $\bar{\gamma}$, to get further insights.

\subsection{Exact Expression}
We obtain the analytical expressions for the exact PDFs of $X$ and $Y$, as detailed in Theorem~\ref{thm1} and Theorem~\ref{thm2}, through which we derive the OP in Lemma~\ref{lm1}.

\begin{proposition}[Exact PDF of $X$]\label{thm1}
The PDF of the sum of double-Rayleigh RVs, $X$, can be expressed in closed-form as
\begin{align}
f_X(x) = \pi^{-\frac{N}{2}} \mathcal{L}^{-1} \left\{ G_{2,2}^{2,2} \left[ \frac{4}{\sigma_\text{S,R}^2 \sigma_\text{R,D}^2 s^2} \Bigg|_{1,1}^{\frac{1}{2},1} \right]^N; x \right\}.
\label{eq:pdf_x}
\end{align}
\end{proposition}

\begin{proof}
To derive the statistics of $X = \sum_{n=1}^{N} |\mathbf{g}(n)| |\mathbf{h}(n)|$, we define $X_n = |\mathbf{g}(n)| |\mathbf{h}(n)|$ as double-Rayleigh RVs, which are the product of two independent Rayleigh RVs. According to \cite[eq.~(3)]{ref33}, the MGF of $X_n$ is given by
\begin{equation}
\mathscr{M}_{X_n}(s) = \frac{1}{\sqrt{\pi}} G_{2,2}^{2,2} \left[ \frac{4}{\sigma_\text{S,R}^2 \sigma_\text{R,D}^2 s^2} \Bigg|_{1,1}^{\frac{1}{2},1} \right].
\label{eq:mgf_xn}
\end{equation}
Since $X$ is the sum of $N$ i.i.d. RVs, i.e., $X_n$, the MGF of $X$ can be obtained by the product of the MGFs of $X_n$s. By applying the inverse Laplace transform, we express the PDF of $X$ as in (\ref{eq:pdf_x}), and thus, the proof is complete.
\end{proof}

\begin{theorem}[Exact PDF of $Y$]\label{thm2}
The exact PDF of $Y$ can be presented in explicit expression as,
\begin{align}
&f_Y(y)= \sum_{m=0}^{\infty} 2y \sigma_\text{I,D}^{2m} \sigma_\text{I,R}^{2N} \sigma_\text{R,D}^{2N} \left[ 
y^{-2m} \frac{\Gamma(m-N-1)}{\Gamma(N-m)\Gamma(m+1)} \right.\notag\\
&\times{}_1F_2\left(N; N-m, N-m; \frac{y^2}{\sigma_\text{I,R}^2 \sigma_\text{R,D}^2}\right) \notag\\
&+ \left({\sigma_\text{I,R}^2 \sigma_\text{R,D}^2}\right)^{N-m-1} \frac{\Gamma(N-m-1)}{\Gamma(N)}\notag\\
& \left.\times {}_1F_2\left(m+1; m-N+2, 1; \frac{y^2}{\sigma_\text{I,R}^2 \sigma_\text{R,D}^2}\right) \right].
\label{eq:pdf_y}
\end{align}

\end{theorem}

\begin{proof}
See Appendix~\ref{appA}.
\end{proof}

\begin{lemma}[Exact Outage Probability]\label{lm1}
The exact OP of RIS-assisted D2D communication including interference at both the destination and the RIS is given in (\ref{Pout_exact}), shown at the top of the next page.
\end{lemma}

\begin{proof}
To characterize the exact OP performance, we begin by deriving the PDFs of $X^2$ and $Y^2$, which are the squares of the RVs $X$ and $Y$, respectively. This derivation is achieved by applying the transformation method to these RVs, utilizing the PDFs of $X$ from (\ref{eq:pdf_x}) and $Y$ from (\ref{eq:pdf_y}). Subsequently, the outcome of this derivation is incorporated into (\ref{eq:pout}), yielding an integral-form expression for OP presented in (\ref{Pout_exact}).\end{proof}
Note that the integrals in (\ref{Pout_exact}) lack readily available closed-form solutions due to the presence of the Meijer G-function and the generalized hypergeometric function, necessitating their evaluation through numerical integration methods. This can be accomplished using well-known mathematical software packages such as MAPLE, MATHEMATICA, and MATLAB.

\begin{figure*}[!t]
\setcounter{equation}{7} 
\begin{align}\label{Pout_exact}
P_{\text{out}}^{\text{exact}} &= \int_{y=0}^{\infty} \sum_{m=0}^{\infty} 2y \sigma_\text{I,D}^{2m} \sigma_\text{I,R}^{2N} \sigma_\text{R,D}^{2N} \left[
y^{-2m} \frac{\Gamma(m-N-1)}{\Gamma(N-m)\Gamma(m+1)} \times \, {}_1F_2\left(N; N-m, N-m; \frac{y^2}{\sigma_\text{I,R}^2 \sigma_\text{R,D}^2}\right) \right. \notag\\
&\,\,\,\,\,\,\,\,\,\,\,\,\,\,\,\,\,\,\,\,\,\,\,\,\,\,\,\,\,\,\,\,\,\,\,\,\,\,\,\,\,\,\,\,+ \left.\left({\sigma_\text{I,R}^2 \sigma_\text{R,D}^2}\right)^{N-m-1} \frac{\Gamma(N-m-1)}{\Gamma(N)}\times \,_1F_2\left(m+1; m-N+2, 1; \frac{y^2}{\sigma_\text{I,R}^2 \sigma_\text{R,D}^2}\right)\right] \notag\\
&\times \int_{x=0}^{y \frac{\gamma_{\text{th}}}{\bar{\gamma}}} \frac{\pi^{N/2}}{2x^{1/2}} L^{-1} \left\{ G_{2,2}^{2,2} \left[\frac{\sigma_\text{S,R}^2 \sigma_\text{R,D}^2}{4 s^2} \bigg|_{0,0.5}^{0,0} \right]^N; \sqrt{x} \right\} \, dx \, dy.
\end{align}
\hrulefill
\end{figure*}
\setcounter{equation}{8} 

\subsection{Approximate Analysis}
As observed in the integral-form expression in (\ref{Pout_exact}), utilizing the exact PDF of $X$ and $Y$, in deriving the OP, leads to computationally complex and intractable analysis. The regular gamma distribution is frequently employed in the existing literature to approximate complex distributions due to its two tunable fading parameters: the shape parameter $k$, and the scale parameter $\theta$. The mean and variance of the gamma distribution are $k\theta$ and $k\theta^2$, respectively. In this subsection, we introduce a sufficiently accurate closed-form approximation for the OP by utilizing the gamma distribution approximation of the PDFs of $X$ and $Y^2$, as established in the following theorems. Then, by applying~(\ref{eq:pout}), we derive closed-form expressions for the OP, as detailed in Lemma~\ref{lm2}.

\begin{theorem}[Approximate PDF of X]\label{thm3}
The PDF of the sum of double-Rayleigh RVs, i.e., $X$, can be approximately expressed using a gamma distribution as
\begin{equation}
f_X(x) \approx \frac{1}{\Gamma(k_X) \theta_X^{k_X}} x^{k_X - 1} e^{-\frac{x}{\theta_X}},
\label{eq:pdf_x_approx}
\end{equation}
with the fading parameters $k_X = \frac{\mathbb{E}[X]^2}{\text{Var}[X]}$, and $\theta_X = \frac{\text{Var}[X]}{\mathbb{E}[X]}$, where $\mathbb{E}[X] = N\frac{\pi}{4} \sigma_\text{S,R} \sigma_\text{R,D}$ and $\text{Var}[X] = N\left(1-\frac{\pi^2}{16} \right)\sigma_\text{S,R}^2 \sigma_\text{R,D}^2$.
\end{theorem}

\begin{proof}
Given that $|\mathbf{g}(n)|$ and $|\mathbf{h}(n)|$ are Rayleigh RVs, the distribution of $X_n = |\mathbf{g}(n)||\mathbf{h}(n)|$ can be approximated as a gamma distribution with the fading parameters $k_{X_n} = \frac{\mathbb{E}[X_n]^2}{\text{Var}[X_n]}$, and $ \theta_{X_n} = \frac{\text{Var}[X_n]}{\mathbb{E}[X_n]}$, where $\mathbb{E}[X_n] = \frac{\pi}{4} \sigma_\text{S,R} \sigma_\text{R,D}$ and $ \text{Var}[X_n] = \left(1 - \frac{\pi^2}{16}\right) \sigma_\text{S,R}^2 \sigma_\text{R,D}^2$ \cite[Lemma~1]{ref34}. According to~\cite{ref35}, as $X = \sum_{n=1}^{N} X_n$ is a sum of $N$ gamma RVs, it follows a new gamma distribution with fading parameters $k_X = Nk_{X_n}$ and $\theta_X = \theta_{X_n}$, leading to (\ref{eq:pdf_x_approx}).
\end{proof}

\begin{theorem}[Approximate PDF of $Y^2$]\label{thm4}
The PDF of $Y' \triangleq Y^2$ can be approximately presented by a gamma distribution as
\begin{equation}
f_{Y'}(y) \approx \frac{1}{\Gamma(k_Y) \theta_Y^{k_Y}} y^{k_Y - 1} e^{-\frac{y}{\theta_Y}}.
\label{eq:pdf_y2_approx}
\end{equation}
The fading parameters are $k_Y = \frac{\mathbb{E}[Y']^2}{\text{Var}[Y']}$, and $\theta_Y = \frac{\text{Var}[Y']}{\mathbb{E}[Y']}$, where the $\mathbb{E}[Y']$ and $\text{Var}[Y'] = \mathbb{E}[Y']^2 - \mathbb{E}[Y'^2]$ can be calculated using
\begin{equation}
\mathbb{E}[Y'] = \sigma_\text{I,D}^2 + N \sigma_\text{I,R}^2 \sigma_\text{R,D}^2,
\label{eq:EY_prime}
\end{equation}
and
\begin{align}
\mathbb{E}[Y'^2] &= 4N \sigma_\text{I,R}^4 \sigma_\text{R,D}^4 + 2N(N-1) \sigma_\text{I,R}^2 \sigma_\text{R,D}^2 \notag\\
&+ 2\sigma_\text{I,D}^2 + 4N \sigma_\text{I,R}^2 \sigma_\text{R,D}^2 \sigma_\text{I,D}^2.
\label{eq:EY_prime2}
\end{align}
\end{theorem}

\begin{proof}
Since $Y'$ is the square of the absolute value of the sum of $N+1$ independent RVs, i.e., $\sum_{n=1}^{N} | \boldsymbol{\beta}(n) | | \boldsymbol{\alpha}(n) | e^{j\boldsymbol{\theta}'(n)} + h_\text{I}$, we apply the central limit theorem due to the large number of terms. This sum is thus approximated as Gaussian, allowing us to approximate the PDF of $Y'$, a squared Gaussian variable, by a gamma distribution \cite{ref_gamma}. Considering the independence of the I-R-D link from the I-D link and $\mathbb{E}[e^{j\varphi}] = 0$, where $\varphi$ is uniformly distributed over $[0,2\pi)$, it can be verified that
\begin{align}
\mathbb{E}[Y'] &= \mathbb{E}\left[\sum_{i=1}^{N} \sum_{j=1}^{N} | \boldsymbol{\beta}(i) | | \boldsymbol{\alpha}(i) | | \boldsymbol{\beta}(j) | | \boldsymbol{\alpha}(j) | e^{j\left(\boldsymbol{\theta}'(i) - \boldsymbol{\theta}'(j)\right)}\right]\notag \\
&+\mathbb{E}[|h_\text{I}|^2].
\label{eq:EY_prime3}
\end{align}

Using (\ref{eq:EY_prime3}) and considering $\mathbb{E}[e^{j\left(\boldsymbol{\theta}'(i) - \boldsymbol{\theta}'(j)\right)}] = 1$ if $i = j$, and $\mathbb{E}[e^{j(\boldsymbol{\theta}'(i) - \boldsymbol{\theta}'(j)}] = 0$ otherwise, we obtain the expectation of $Y'$ as in (\ref{eq:EY_prime}). Subsequently,
\begin{align}
&\mathbb{E}[Y'^2]= \mathbb{E}\left[ \left| \sum_{n=1}^{N} | \boldsymbol{\beta}(n) | | \boldsymbol{\alpha}(n) | e^{j\theta_n'} + h_\text{I} \right|^4 \right] \nonumber \notag\\
&= 4\mathbb{E}\left[\sum_{i=1}^{N} \sum_{j=1}^{N} | \boldsymbol{\beta}(i) | | \boldsymbol{\alpha}(i) | | \boldsymbol{\beta}(j) | | \boldsymbol{\alpha}(j) | e^{j\left(\boldsymbol{\theta}'(i) - \boldsymbol{\theta}'(j)\right)}\right] \mathbb{E}[|h_\text{I}|^2] \notag\\
&+ \mathbb{E}\left[ \left| \sum_{i=1}^{N} \sum_{j=1}^{N} \sum_{p=1}^{N} \sum_{q=1}^{N} | \boldsymbol{\beta}(i) | | \boldsymbol{\alpha}(i) | | \beta(j) | | \boldsymbol{\alpha}(j) | | \beta(p) |\right.\right.\notag\\
&\times\left.\left. | \boldsymbol{\alpha}(p) | | \beta(q) | | \boldsymbol{\alpha}(q) |
e^{j\left(\boldsymbol{\theta}'(i) - \boldsymbol{\theta}'(j) + \boldsymbol{\theta}'(p) - \boldsymbol{\theta}'(q)\right)}\right|^4 \right] \nonumber \\
&+ \mathbb{E}[|h_\text{I}|^4].
\label{eq:EY_prime4}
\end{align}
Applying $\mathbb{E}[e^{j\left(\boldsymbol{\theta}'(i) - \boldsymbol{\theta}'(j) + \boldsymbol{\theta}'(p) - \boldsymbol{\theta}'(q)\right)}] = 1$ if $i = j$ and $p = q$ (or $i = q$ and $p = j$), and $\mathbb{E}[e^{j(\boldsymbol{\theta}'(i) - \boldsymbol{\theta}'(j) + \boldsymbol{\theta}'(p) - \boldsymbol{\theta}'(q))}] = 0$ otherwise, we obtain (\ref{eq:EY_prime2}). This completes the proof.
\end{proof}

\begin{lemma}[Approximate Outage Probability]\label{lm2}
The approximate OP for RIS-assisted D2D communication, in the presence of interference at both the destination and the RIS, can be derived in a closed-form expression as
\begin{align}
P_{\text{out}}^{\text{approx}}&= 1 -\sum_{i=0}^{k_X - 1} \frac{1}{i!}\frac{\Gamma(2k_Y + i)}{2^{(k_Y + \frac{i}{2}) - 1} \Gamma(k_Y)} \left(\frac{\gamma_{\text{th}} \theta_Y}{\bar{\gamma} \theta_X^2}\right)^{\frac{i}{2}} \notag \\
& \times \exp\left(\frac{\gamma_{\text{th}} \theta_Y}{8\bar{\gamma} \theta_X^2}\right) D_{-(2k_Y + i)}\left(\sqrt{\frac{\gamma_{\text{th}} \theta_Y}{2\bar{\gamma} \theta_X^2}}\right).
\label{eq:pout_approx}
\end{align}
\end{lemma}
\begin{proof}
Detailed in Appendix~\ref{appB}.
\end{proof}

It should be noted that (\ref{eq:pout_approx}) expresses the closed-form representation of OP as a sum of well-defined parabolic cylinder functions that can be directly evaluated in several software packages, such as MAPLE, MATHEMATICA, MATLAB, etc.

\subsection{Asymptotic Analysis}
To comprehend the impact of key system parameters on OP, we investigate the system's asymptotic performance in a high average SIR regime, enabling us to derive the diversity order and the coding gain.
\begin{lemma}[Asymptotic Outage Probability]\label{lm3}
The asymptotic expression of the OP for high average SIRs is given by
\begin{equation}
P_{\text{out}}^{\text{asymp}} = \frac{\Gamma(k_Y + \frac{k_X}{2})}{\Gamma(k_X + 1) \Gamma(k_Y)} \left( \frac{\bar{\gamma}_\text{S} \theta_Y}{\bar{\gamma}_\text{I} \gamma_{\text{th}} \theta_X^2} \right)^{-\frac{k_X}{2}}.
\label{eq:pout_asymp}
\end{equation}
\end{lemma}
\begin{proof}
Detailed in Appendix~\ref{appC}.
\end{proof}

As it can be seen in (\ref{eq:pout_asymp}), an increase in the average INR, i.e., $\bar{\gamma}_\text{I}$, leads to higher OP, consequently reducing the system performance. Conversely, an enhancement in the average SNR, i.e., $\bar{\gamma}_\text{S}$, improves the system performance. Proposition~\ref{pro1} provides insights into the diversity order of the considered RIS-assisted D2D system.

\begin{proposition}\label{pro1}
The diversity order and coding gains of RIS-assisted D2D communications in the presence of interference are
\begin{equation}
\mathcal{G}_d = \frac{N}{4} \frac{\pi^2}{(16 - \pi^2)},
\label{eq:diversity_order}
\end{equation}
and
\begin{equation}
\mathcal{G}_c = \left( \frac{\Gamma(k_Y + \frac{k_X}{2})}{\Gamma(k_X + 1) \Gamma(k_Y)} \right)^{-\frac{2}{k_X}} \frac{\theta_Y}{\theta_X^2},
\label{eq:coding_gain}
\end{equation}
respectively.
\end{proposition}

\begin{proof}
The diversity order is defined by the negative slope of the OP plotted against the average SIR in a log-log scale \cite{ref37}. Using (\ref{eq:pout_asymp}), we calculate the diversity order as
\begin{equation}
\mathcal{G}_d = \lim_{\bar{\gamma} \to \infty} \left( -\frac{\log P_{\text{out}}^{\text{asymp}}}{\log \bar{\gamma}} \right) = \frac{k_X}{2},
\label{eq:diversity_order_calc}
\end{equation}
which according to (\ref{eq:pdf_x_approx}), $k_X$ is expressed as $k_X = \frac{N \pi^2}{16 - \pi^2}$, and the diversity order is derived as (\ref{eq:diversity_order}). In general, the OP in the high SNR regime can be asymptotically approximated as $P_{\text{out}} = (\mathcal{G}_c \bar{\gamma})^{-\mathcal{G}_d}$ \cite{ref38}. This formulation indicates that the coding gain, i.e., $\mathcal{G}_c$, can be derived using~(\ref{eq:pout_asymp}) as (\ref{eq:coding_gain}). This concludes the proof.
\end{proof}

According to (\ref{eq:diversity_order}), the system's diversity order solely depends on $N$ and shows a linear increase in relation to it.\footnote{This is identical to the diversity order reported in \cite[eq.~(43)]{ref39} for scenarios without interference.} Additionally, this order remains unaffected by the interference, but the presence of interference signals can degrade the system performance by reducing the coding gain.\footnote{This observation was also reported in previous literature that neglected the reflected interference by the RIS \cite{ref5}.}
\begin{figure}[!t]
\centering
\includegraphics[width=2.5in]{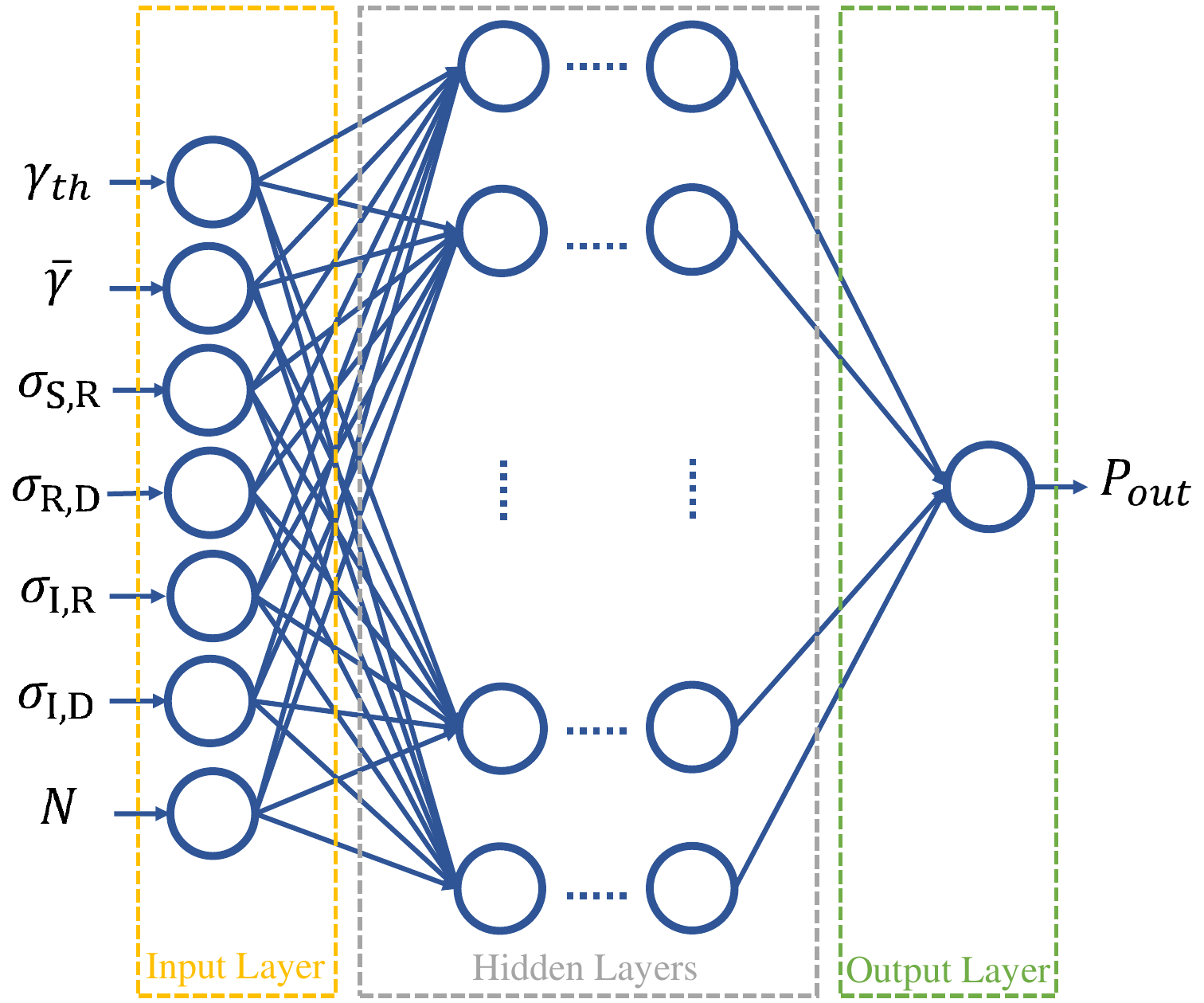}
\caption{The DNN structure for outage probability prediction.}
\label{fig_2}
\end{figure}

\section{Outage Probability Prediction} \label{sec:dnn_approach}
To reduce the complexity and enable real-time analysis of the OP performance, this section redefines the OP evaluation as a regression problem. Subsequently, we propose an OP prediction method based on a supervised DNN.

\subsection{Dataset Generation}
For the OP corresponding to (\ref{Pout_exact}), it is identified that parameters such as $\gamma_{\text{th}}$, $\bar{\gamma}$, $\sigma_\text{S,R}$, $\sigma_\text{R,D}$, $\sigma_\text{I,R}$, $\sigma_\text{I,D}$, and $N$ influence the OP performance. The primary objective is to enable ML approaches to learn the relationship between these input parameters, represented as $\textbf{input} = [\gamma_{\text{th}}, \bar{\gamma}, \sigma_\text{S,R}, \sigma_\text{R,D}, \sigma_\text{I,R}, \sigma_\text{I,D}, N]$, and the corresponding exact OP given by (\ref{Pout_exact}), represented as $\text{output}$. Specifically, the dataset, denoted as $\mathcal{D} = \{\mathbf{D}(1), \mathbf{D}(2), \ldots, \mathbf{D}(S)\}$, consists of pairs $\mathbf{D}(s) = \{\textbf{input}(s), \text{output}(s)\}$, $\forall s = 1, \ldots, S$. A total of $S = 10,000$ data points are generated for this dataset. In each independent run, these input parameters are determined as the $s$-th inputs, and the exact OP performance, i.e., $P_{\text{out}}^{\text{exact}}$, is calculated using (\ref{Pout_exact}), serving as the $s$-th training target. The dataset, i.e., $\mathcal{D}$, is then divided into three subsets: $70\%$ for training, $10\%$ for testing, and the remaining~$20\%$ reserved as a validation set to assess the generalization ability of the trained ML model.

\subsection{Network Structure}
Among various ML approaches suitable for our regression problem, some candidates have been selected, and then, by further consideration, finally DNNs have been adopted. The detailed consideration of neural network structures is not addressed here, as it is beyond the scope of this work. The DNN structure used in this study, as illustrated in Fig.~\ref{fig_2}, comprises an input layer, three fully connected hidden layers with 20, 30, and 20 neurons, respectively, and an output layer.\footnote{The configuration of the DNN was determined through trial-and-error.} The input passes through these hidden layers, and the predicted output, i.e., ${\hat{P}}_{out}$, is generated.

\begin{figure}[!t]
\centering
\includegraphics[width=3in]{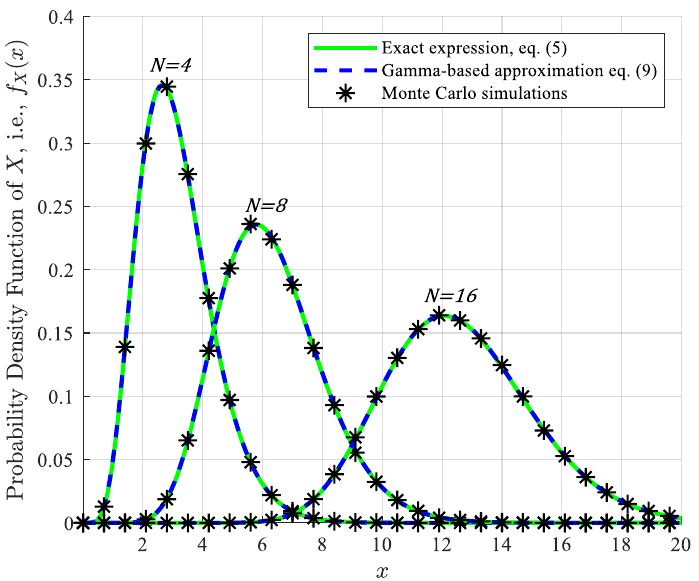}
\caption{Comparison of Monte Carlo simulation results with the exact expression from Theorem~\ref{thm1} and the gamma-based approximate expression from Theorem~\ref{thm3} for the PDFs of double-Rayleigh RVs, i.e., $X$.}
\label{fig_3}
\end{figure}

\subsection{Training loss function}
To minimize the difference between the predicted OP, $\hat{P}_{\text{out}}$, and the actual OP, $P_{\text{out}}^{\text{exact}}$, we employ the mean square error (MSE) loss function:
\begin{equation}
\text{MSE} = \frac{1}{S} \sum_{s=1}^{S} \left( \hat{P}_{\text{out}}(s) - P_{\text{out}}^{\text{exact}}(s) \right)^2.
\label{eq:mse}
\end{equation}

In each epoch, the optimizer iteratively optimizes the weight in each layer using the Levenberg-Marquardt method \cite{ref40}, based on the loss value.

It is important to note that the ML training stage requires a powerful computational server and is conducted offline, whereas the implementation stage can be executed online. Furthermore, the trained ML model does not require updates, even if there are changes in the users' positions. Retraining is only necessary when the considered interference model changes.

\section{Simulation Results} \label{sec:simulation_results}
This section presents numerical results obtained from Monte Carlo simulations, which serve to verify the accuracy of the analytical expressions derived for the PDF, OP, and the diversity order of RIS-assisted interference-limited systems, as discussed in previous sections. These simulations also explore the impact of interference on the system performance and demonstrate real-time analysis of the OP performance using the proposed DNN-based OP prediction. We conduct Monte Carlo simulations with 10,000 realizations for each data point. For simplicity, the threshold $\gamma_{\text{th}}$ is set to 0~dB, and all channels are assumed to have unity variance, i.e., $\sigma_\text{S,R} = \sigma_\text{R,D} = \sigma_\text{I,R} = \sigma_\text{I,D} = 1$.

\begin{figure}[!t]
\centering
\includegraphics[width=3in]{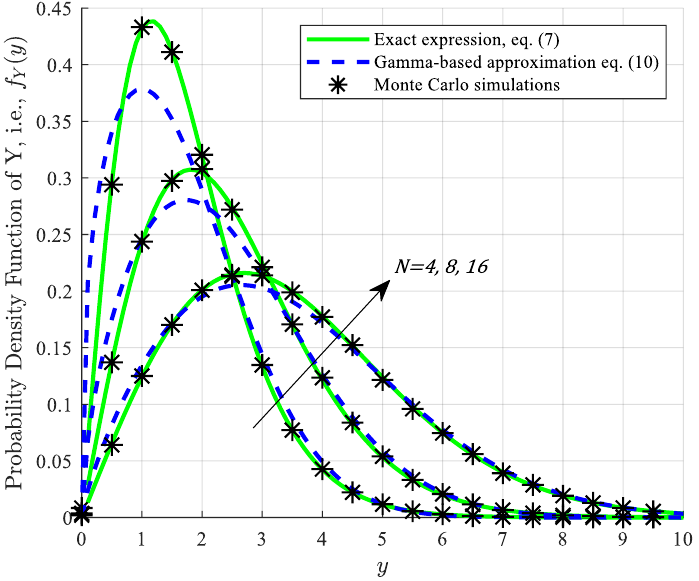}
\caption{Comparison of the exact and approximate PDFs of $Y$ from Theorem~\ref{thm2} and Theorem~\ref{thm4}, versus Monte Carlo simulation results for various values of $N$.}
\label{fig_4}
\end{figure}
\subsection{Probability Density Function}
Fig.~\ref{fig_3} illustrates a comparison between the exact expression from Theorem~\ref{thm1}, i.e., (\ref{eq:pdf_x}), the gamma-based approximate expression from Theorem~\ref{thm3}, i.e., (\ref{eq:pdf_x_approx}), and the Monte Carlo simulation results for the sum of double-Rayleigh RVs, i.e., $X$. The figure shows a perfect match between the Monte Carlo simulations and both the exact and asymptotic PDFs derived in (\ref{eq:pdf_x}) and (\ref{eq:pdf_x_approx}), confirming the accuracy of our derived expressions. Additionally, it is observed that an increase in the number of reflecting elements, i.e., $N$, leads to improved fading conditions.

Fig.~\ref{fig_4} validates the derived PDFs for $Y$ by comparing the exact and approximate expressions from Theorem~\ref{thm2}, i.e., (\ref{eq:pdf_y}), and Theorem~\ref{thm4}, i.e., (\ref{eq:pdf_y2_approx}), with Monte Carlo simulation results for $N \in \{4, 8, 16\}$. This figure indicates a precise match of the exact expression with the simulation results across all values of $N$, with good agreement between the exact and approximated PDFs over a broad range of $y$ values. The figure further reveals that the accuracy of the approximated PDFs largely depends on the number of reflecting elements, achieving an accurate approximation with the gamma-based approach when $N \geq 16$.

\subsection{Outage Probability }
Fig.~\ref{fig_5} presents the OP performance as a function of the average SNR, i.e., $\bar{\gamma}_\text{S}$, for various values of $N$ when the average INR is set at $\bar{\gamma}_\text{I} = \{0, 15\}$~dB. The numerical results, derived through Monte Carlo simulations, align well with our analytical findings, including the exact expression from Lemma~\ref{lm1} and the gamma-based approximate expression from Lemma~\ref{lm2}. The figure demonstrates that, in the high SNR region, the exact values, converge to the asymptotic results from Lemma~\ref{lm3}, i.e., (\ref{eq:pout_asymp}). Although by increasing $N$, the interference through the RIS is increased, the phase shift response of each element is optimized based on the desired signal. Consequently, as $N$ increases, the enhancement of the desired signal outweighs the increase in the interference, effectively mitigating the impacts of the interference. For example, at an average INR of $\bar{\gamma}_\text{I} = 0$, to achieve an OP of $10^{-5}$, $N = 4$ requires an SNR of 20~dB, while $N = 8$ requires only 10~dB. Moreover, Fig.~\ref{fig_5} validates the accuracy of the diversity analysis presented in Proposition~\ref{pro1}. The figure clearly demonstrates that, for a given value of $N$, the asymptotic curves exhibit identical slopes. This observation highlights that the diversity order $\mathcal{G}_d$ is solely determined by $N$, reinforcing the conclusion that the presence of the interference does not impact the system’s diversity order.

\begin{figure}[!t]
\centering
\includegraphics[width=2.9in]{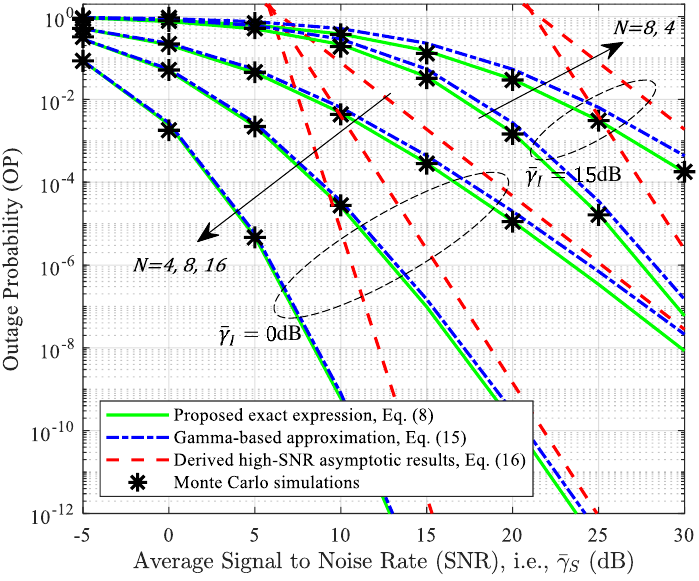}
\caption{Outage probability versus the average SNR, i.e., $\bar{\gamma}_\text{S}$, for various values of $N$ when the average INR $\bar{\gamma}_\text{I} = \{0, 15\}$~dB.}
\label{fig_5}
\end{figure}

To illustrate the adverse impacts of the interference on the RIS-assisted communication, Fig.~\ref{fig_6} plots the OP curves versus average INR, i.e., $\bar{\gamma}_\text{I}$, with $\bar{\gamma}_\text{S} = 15$~dB for $N \in \{4, 8, 16\}$. These results validate the expressions presented in (\ref{Pout_exact}) and (\ref{eq:pout_approx}). It is evident that an increase in $\bar{\gamma}_\text{I}$ leads to OP performance degradation. Moreover, Fig.~\ref{fig_6} indicates that increasing~$N$ significantly reduces the adverse effect of the interference on the OP performance. For instance, an INR of~5~dB results in an OP of $10^{-2}$ for $N = 4$, while in the case of $N = 8$ the OP improves to a significantly lower $10^{-9}$.
Finally, Fig.~\ref{fig_6} demonstrates that increasing interference, i.e., INR, significantly reduces OP performance. For instance, for $N=8$, an INR of 5~dB results in an OP of approximately $10^{-5}$, while at an INR of -5~dB, the OP improves significantly to approximately $10^{-10}$. This underscores the importance of considering interference in the performance evaluation of RIS-assisted systems.

\subsection{Performance Analysis of the DNN-based OP prediction}
This subsection presents simulation results to validate the efficacy and computational advantages of the proposed DNN-based OP prediction method. For a fair comparison, all measurements are conducted on computers with identical configurations.
\begin{figure}[!t]
\centering
\includegraphics[width=2.9in]{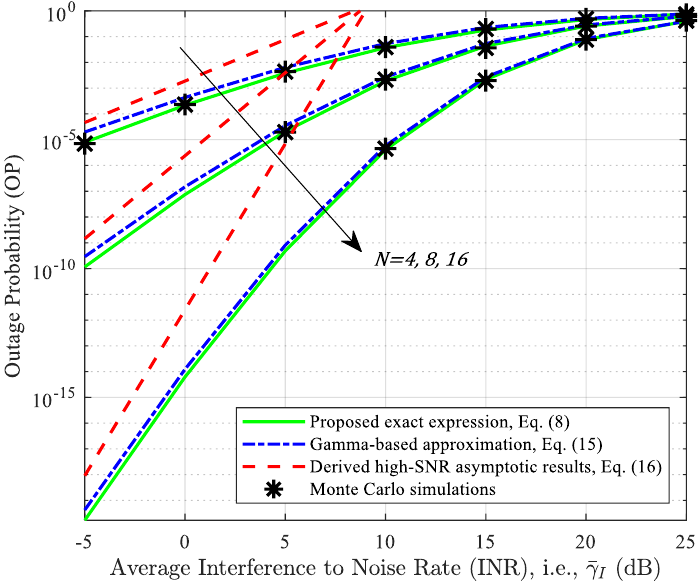}
\caption{Outage probability as a function of average INR, i.e., $\bar{\gamma}_\text{I}$, with a fixed SNR of $\bar{\gamma}_\text{S} = 15$~dB, for $N \in \{4, 8, 16\}$.}
\label{fig_6}
\end{figure}

\begin{table}[!b]
\centering
\caption{Exclusion time and MSE in the case of 1000 samples using analytical OP expressions, and DNN-based OP prediction approach.}
\label{tab:label1}
\begin{tabular}{cccc}
\toprule
\textbf{Approach} & \multicolumn{1}{p{1.5cm}}{\centering\textbf{Exact\\ expression}} & \multicolumn{1}{p{2cm}}{\centering\textbf{Gamma-based\\ expression}} & \multicolumn{1}{p{1.5cm}}{\centering\textbf{DNN-based\\ prediction}} \\
\midrule
MSE                & 0                         & $4.66 \times 10^{-5}$            & $9.37 \times 10^{-6}$           \\
Time (sec)         & 1,284.2                   & 308.336                         & 0.247 \\
\bottomrule
\end{tabular}
\end{table}

In Table~\ref{tab:label1}, we summarize the execution time and MSE values of both the DNN-based OP prediction approach and the analytical expressions derived for the OP, i.e., the exact integral-form expression in (\ref{Pout_exact}) and the gamma-based approximation expression in (\ref{eq:pout_approx}). 
From Table~\ref{tab:label1}, the DNN-based OP prediction approach outperforms the analytical expressions derived for the OP in terms of time efficiency, albeit with a slight loss in accuracy. 
The exact integral-form expression in (\ref{Pout_exact}) has a computational complexity of $\mathcal{O}\left(M N^2_\text{int} N^3_\text{MGfunc}\right)$, where $N_\text{int}$ denotes the number of subdivisions in the integration interval for numerical calculations, and $N_\text{MGfunc}$ represents the complexity associated with calculating the Meijer G-function \cite{Complexity1}. Similarly, the gamma-based approximation expression in (\ref{eq:pout_approx}) has a computational complexity of $\mathcal{O}\left(k_X N^3_\text{PCfunc}\right)$, where $N_\text{PCfunc}$ represents the complexity associated with calculating the Parabolic Cylinder function, depending on the method used such as integral representations, Maclaurin series, and confluent hypergeometric functions \cite{Complexity2}. Both of these complexities can be considered impractical for many applications due to their potentially high computational demands.
Thanks to the offline training procedure, the DNN-based OP prediction approach significantly reduces execution time compared to the analytical expressions, without a considerable increase in MSE. Specifically, the computational complexity decreases to $\mathcal{O} (1)$ for DNN, since the network fixed layer sizes and a fixed input size. This demonstrates a substantial computational advantage of employing machine learning for OP performance prediction in RIS-assisted D2D communication systems.

\begin{figure}[!t]
\centering
\includegraphics[width=3in]{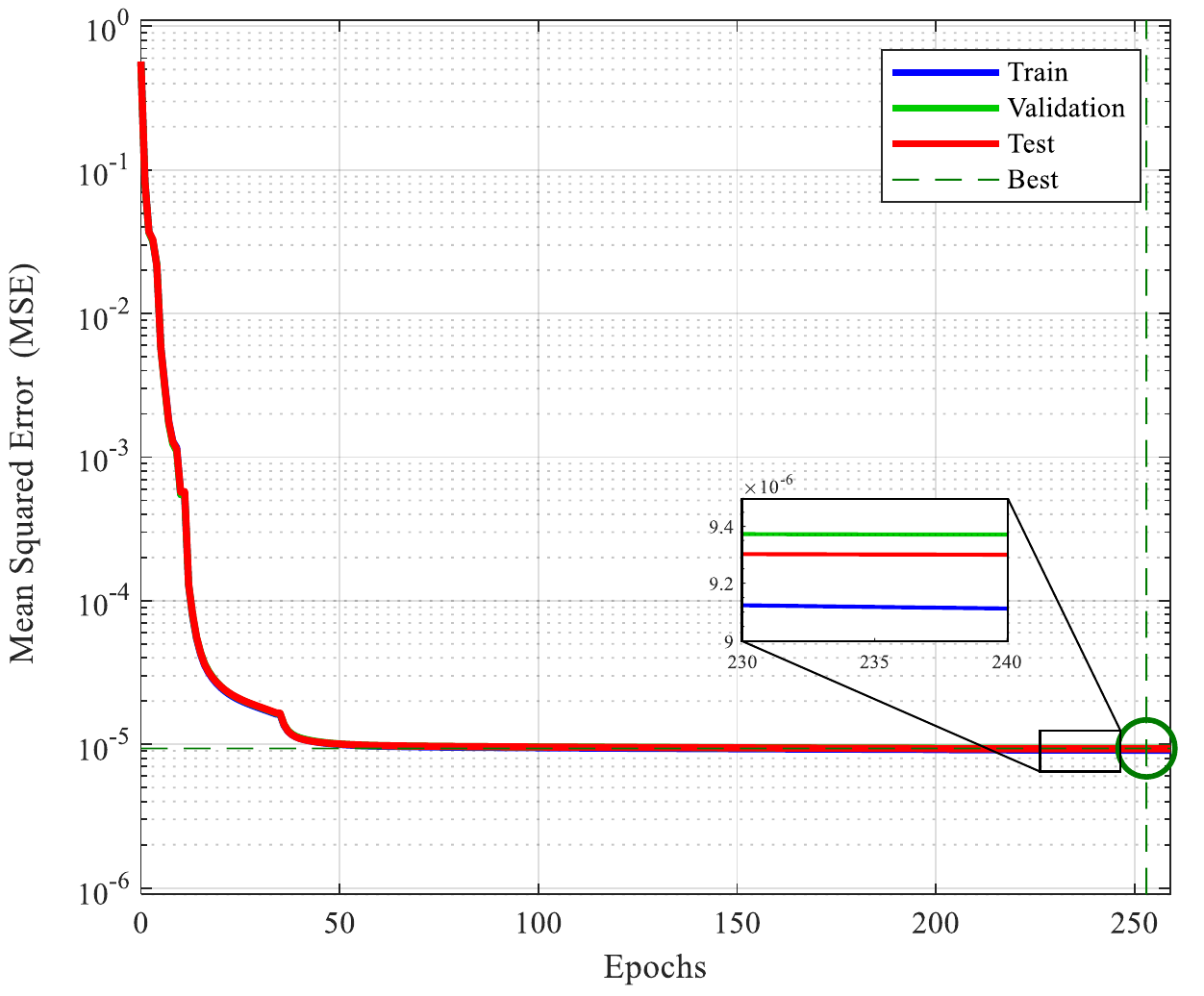}
\caption{Training process of the proposed DNN-based OP prediction method.}
\label{fig_7}
\end{figure}
Moreover, Fig.~\ref{fig_7} illustrates the training process of the DNN-based OP prediction model. As the number of epochs increases, the model demonstrates a high level of precision, evidenced by the lowest validation MSE of $9.37 \times 10^{-6}$ achieved at epoch 253. The cessation of training at epoch 259, prompted by a rise in MSE following its lowest point at epoch 253, reflects a well-implemented early stopping mechanism aimed at preventing overfitting. The training behavior depicted in Fig.~\ref{fig_7} suggests that the proposed DNN-based model has successfully learned the underlying patterns of the dataset and can accurately predict the system’s OP performance.

Finally, Fig.~\ref{fig_8} presents the results of regression analysis. The value of $R$, indicative of the relationship between the predicted and actual results, shows the predictive accuracy of the DNN-based approach. An $R$ value of 0.99996 signifies exceptional prediction performance of the DNN-based approach, underscoring its effectiveness in this application.

\section{Conclusion} \label{sec:conclusion}
In this paper, we have extensively analyzed the impact of co-channel interference on RIS-assisted D2D communication systems, offering a novel perspective by considering interference at both the user and the RIS. We introduced an integral-form expression for calculating the OP and a gamma-based approximation, offering a practical closed-form solution. Furthermore, we developed a DNN-based approach for real-time OP prediction, demonstrating its substantial computational efficiency and accuracy. Our investigation into the diversity order and coding gain in the presence of interference reveals that while the interference does not affect the diversity order of the system, it significantly degrades the performance by reducing the coding gain. Additionally, our results indicate that increasing the number of RIS elements effectively mitigates the adverse effects of the interference, leading to an overall enhancement in the system performance. In particular, the employment of ML in this study not only highlights the potential role of ML in wireless communications but also opens new avenues for employing ML-based computational techniques in real-time performance analysis of communication systems.
\begin{figure}[!t]
\centering
\includegraphics[width=3in]{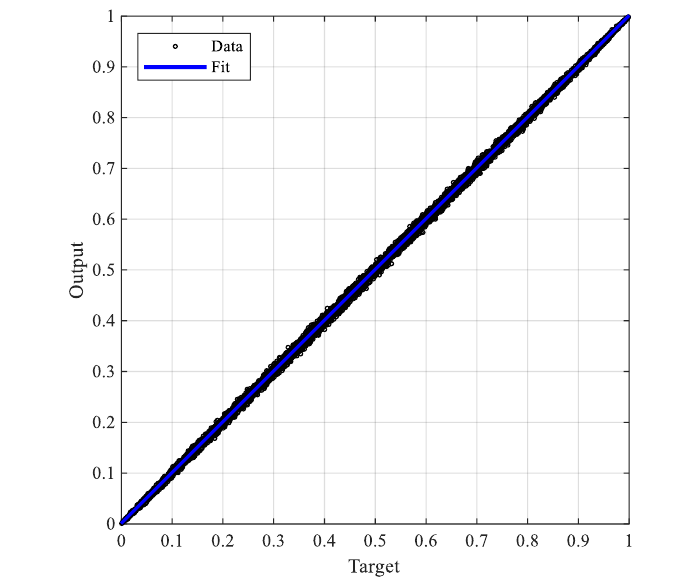}
\caption{Regression analysis results for the proposed DNN-based OP prediction.}
\label{fig_8}
\end{figure}

\section*{Acknowledgments}
This work was supported by the Ministry of Science and Higher Education of the Republic of Kazakhstan, under project No. AP13068587, and by Iran National Science Foundation (INSF) under project No. 4001804.

{\appendices
\section{Proof of Theorem~\ref{thm2}: Exact PDF of $Y$}\label{appA}
Let $Y = \left| \sum_{n=1}^{N+1} \boldsymbol{\mathcal{h}}(n) e^{j \boldsymbol{\phi}(n)} \right|$, where $\boldsymbol{\mathcal{h}}(n)$ for $n \in \mathscr{N}$ is defined as $\boldsymbol{\mathcal{h}}(n) = |\boldsymbol{\beta}(n)| |\boldsymbol{\alpha}(n)|$, which follows double-Rayleigh RVs. For $n=N+1$, $\boldsymbol{\mathcal{h}}(n)$ is defined as $\boldsymbol{\mathcal{h}}(N+1) = |h_\text{I}|$ and follows a Rayleigh RV. Moreover, $\boldsymbol{\phi}(n) = \boldsymbol{\theta}'(n)$ for $n \in \mathscr{N}$ and $\boldsymbol{\phi}(N+1) = \angle h_\text{I}$. According to \cite{ref41, ref_R1}, the PDF of $Y$ can be expressed as
\begin{equation}
f_Y(y) = y \int_{0}^{\infty} \rho \mathcal{I}_0(y\rho) \Lambda(\rho) \, d\rho,
\label{eq:pdf_Y}
\end{equation}
where $\Lambda(\rho)$ is given by
\begin{align}
\Lambda(\rho) &= \mathbb{E}_{\boldsymbol{\mathcal{h}}(1),\ldots,\boldsymbol{\mathcal{h}}(N+1)} \left[ \prod_{n=1}^{N+1} \mathcal{I}_0(\boldsymbol{\mathcal{h}}(n)\rho) \right] \notag\\
&= \prod_{n=1}^{N+1} \mathbb{E}_{\boldsymbol{\mathcal{h}}(n)} \left[ \mathcal{I}_0(\boldsymbol{\mathcal{h}}(n)\rho) \right].
\label{eq:lambda_rho}
\end{align}
The last equality results from the independency of $\boldsymbol{\mathcal{h}}(1),\ldots,\boldsymbol{\mathcal{h}}(N+1)$. The expected value in (\ref{eq:lambda_rho}) can be derived as
\begin{equation}
\mathbb{E}_{\boldsymbol{\mathcal{h}}(n)} \left[ \mathcal{I}_0(\boldsymbol{\mathcal{h}}(n)\rho) \right] = \int_{0}^{\infty} \mathcal{I}_0(h\rho) f_{\boldsymbol{\mathcal{h}}(n)}(h) \, dh.
\label{eq:expected_value_h}
\end{equation}
The PDF of $\boldsymbol{\mathcal{h}}(n)$, for $n \in \mathscr{N}$, which are the PDF of double-Rayleigh RVs \cite[eq.~(6)]{ref42}, is given by
\begin{equation}
f_{\boldsymbol{\mathcal{h}}(n)}(h) = \frac{4h}{\sigma_\text{I,R}^2 \sigma_\text{R,D}^2} \mathcal{K}_0\left(\frac{2h}{\sigma_\text{I,R} \sigma_\text{R,D}}\right),\ n \in \mathscr{N},
\label{eq:pdf_h_n}
\end{equation}
Transforming the second order Bessel functions into Meijer’s G-functions \cite[03.04.26.0037.01]{ref43}, can be expressed as
\begin{equation}
f_{\boldsymbol{\mathcal{h}}(n)}(h) = \frac{2h}{\sigma_\text{I,R}^2 \sigma_\text{R,D}^2} G_{0,2}^{2,0}\left(\frac{h}{\sigma_\text{I,R} \sigma_\text{R,D}},\frac{1}{2} \bigg| 0,0\right),\ n \in \mathscr{N}.
\label{eq:pdf_h_n_meijer}
\end{equation}
Also, by transforming the first-order Bessel functions into Meijer’s G-functions \cite[03.01.26.0056.01]{ref43}, we have
\begin{equation}
\mathcal{I}_0(h\rho) = G_{0,2}^{1,0}\left(\frac{h\rho}{2},\frac{1}{2} \bigg| 0,0\right).
\label{eq:I0_meijer}
\end{equation}
By substituting (\ref{eq:pdf_h_n_meijer}) and (\ref{eq:I0_meijer}) into (\ref{eq:expected_value_h}) and utilizing \cite{ref44} and \cite[07.23.17.0046.01]{ref43}, equation (\ref{eq:expected_value_h}) for $n \in \mathscr{N}$ can be written as,
\begin{equation}
\mathbb{E}_{\boldsymbol{\mathcal{h}}(n)}\left[\mathcal{I}_0(\boldsymbol{\mathcal{h}}(n)\rho)\right] = \frac{4}{\rho^2 \sigma_\text{I,R}^2 \sigma_\text{R,D}^2 + 4},\ n \in \mathscr{N}.
\label{eq:expected_value_h_meijer}
\end{equation}

The PDF of $\boldsymbol{\mathcal{h}}(N+1)$, which is the Rayleigh RVs, can be expressed as
\begin{equation}
f_{\boldsymbol{\mathcal{h}}(N+1)}(h) = \frac{2h}{{\sigma_\text{I,D}}^2} \exp{\left(-\frac{h^2}{{\sigma_\text{I,D}}^2}\right)}.
\label{eq:pdf_h_N_plus_1}
\end{equation}
By transforming the Bessel functions into series, \(\mathcal{I}_0(h\rho) = \sum_{m=0}^{\infty} \frac{1}{m! \Gamma(m+1)} \left(\frac{h\rho}{2}\right)^{2m}\) \cite[03.02.02.0001.01]{ref43}, equation (23) for $n = N+1$, can be expressed as
\begin{align}
\mathbb{E}_{\boldsymbol{\mathcal{h}}(n)}&\left[\mathcal{I}_0(\boldsymbol{\mathcal{h}}(n)\rho)\right]= \sum_{m=0}^{\infty} \frac{\left(\frac{\rho}{2}\right)^{2m}}{m! \Gamma(m+1)}\notag\\
&\times\int_{0}^{\infty} \frac{2h^{2m+1}}{{\sigma_\text{I,D}}^2} \exp\left(-\frac{h^2}{{\sigma_\text{I,D}}^2}\right) \, dh.
\label{eq:expected_value_h_series}
\end{align}
Utilizing \cite[eq.~(3.461.3)]{ref25}, the integral of (23) can be evaluated as \(\int_{0}^{\infty} \frac{2h^{2m+1}}{{\sigma_\text{I,D}}^2} \exp\left(-\frac{h^2}{2{\sigma_\text{I,D}}^2}\right) dh = \sigma_\text{I,D}^{2m} \Gamma(m+1)\), and then, (\ref{eq:expected_value_h_series}) reduces to
\begin{equation}
\mathbb{E}_{\boldsymbol{\mathcal{h}}(N+1)}\left[\mathcal{I}_0(\boldsymbol{\mathcal{h}}(N+1)\rho)\right] = \sum_{m=0}^{\infty} \frac{\left(\frac{\sigma_\text{I,D}\rho}{2}\right)^{2m}}{\Gamma(m+1)}.
\label{eq:expected_value_h_N_plus_1}
\end{equation}
By substituting (\ref{eq:expected_value_h_meijer}) and (\ref{eq:expected_value_h_N_plus_1}) into (\ref{eq:lambda_rho}), we have
\begin{equation}
\Lambda(\rho) = \left(\frac{4}{4 + \sigma_\text{I,R}^2 \sigma_\text{R,D}^2 \rho^2}\right)^N \sum_{m=0}^{\infty} \frac{\left(\frac{\sigma_\text{I,D}\rho}{2}\right)^{2m}}{\Gamma(m+1)},
\label{eq:lambda_rho_series}
\end{equation}
and then by substituting (\ref{eq:lambda_rho_series}) into (\ref{eq:pdf_Y}), the PDF of $Y$ can be expressed as
\begin{equation}
f_Y(y) = y \sum_{m=0}^{\infty} \frac{{\sigma_\text{I,D}}^{2m} \sigma_\text{I,R}^{2N} \sigma_\text{R,D}^{2N}}{\Gamma(m+1) 2^{2m - 2N}} \int_{0}^{\infty} \frac{\rho^{2m+1} \mathcal{I}_0(y\rho)}{\left(\rho^2 + \frac{4}{\sigma_\text{I,R}^2 \sigma_\text{R,D}^2}\right)^N} \, d\rho.
\label{eq:pdf_Y_series}
\end{equation}
The proof is completed by utilizing \cite[eq.~(6.565.8)]{ref25}, followed by some algebraic manipulations to obtain the PDF of $Y$ in closed-form as in (\ref{eq:pdf_y}).

\section{Proof of Lemma~\ref{lm2}: Approximate Outage Probability}\label{appB}
The derivation of Lemma 2 involves incorporating the approximate PDFs of $X^2$ and $Y^2$, as derived from (\ref{eq:pdf_x_approx}) and (\ref{eq:pdf_y2_approx}), into (\ref{eq:pout}), which yields
\begin{align}
P_{\text{out}}^{\text{approx}} &= \int_{y=0}^{\infty} \frac{1}{\Gamma(k_Y) \theta_Y^{k_Y}} y^{k_Y - 1} e^{-\frac{y}{\theta_Y}}\notag\\
&\times\int_{x=0}^{y \frac{\bar{\gamma}_\text{I} \gamma_{\text{th}}}{\bar{\gamma}_\text{S}}} \frac{1}{2 \Gamma(k_X) \theta_X^{k_X}} x^{\frac{k_X}{2} - 1} e^{-\frac{\sqrt{x}}{\theta_X}} \, dx \, dy.
\label{eq:pout_approx_integral}
\end{align}
Applying the integral results from \cite[eq.~(3.381.8)]{ref25}, (\ref{eq:pout_approx_integral}) can be rewritten as
\begin{align}
P_{\text{out}}^{\text{approx}} &= \int_{y=0}^{\infty} \frac{1}{\Gamma(k_Y) \theta_Y^{k_Y}} y^{k_Y - 1} e^{-\frac{y}{\theta_Y}}\notag\\
&\times\frac{1}{\Gamma(k_X)} \gamma\left(k_X, \frac{\left(y \frac{\bar{\gamma}_\text{I} \gamma_{\text{th}}}{\bar{\gamma}_\text{S}}\right)^{\frac{1}{2}}}{\theta_X}\right) \, dy.
\label{eq:pout_approx_rewritten}
\end{align}
Considering a series expansion of the cumulative distribution function (CDF) of the gamma distribution, we obtain \cite{ref36}
\begin{equation}
\frac{\gamma(\alpha, \beta)}{\Gamma(\alpha)} = 1 - \sum_{i=0}^{\alpha - 1} \frac{\beta^i}{i!} e^{-\beta}.
\label{eq:gamma_cdf_series}
\end{equation}
Utilizing the integral results from \cite[eq.~(3.462.1)]{ref25}:
\begin{align}
\int_{0}^{\infty} y^{\frac{\nu}{2} - 1} &e^{-\beta y - \gamma y^{\frac{1}{2}}} \, dx \notag\\
&= 2 (2\beta)^{-\frac{\nu}{2}} \Gamma(\nu) \exp\left(\frac{\gamma^2}{8\beta}\right) D_{-\nu}\left(\frac{\gamma}{\sqrt{2\beta}}\right),
\label{eq:integral_result}
\end{align}
we can deduce the approximate OP as presented in (\ref{eq:pout_approx}). Thus, the proof is complete.

\section{Proof of Lemma~\ref{lm3}: Asymptotic Outage Probability}\label{appC}
By considering a high average SIR, and applying an asymptotic expression of $D_a(x)$ at $x \rightarrow 0$ \cite[eq.~(07.41.06.0010.01)]{ref43}, (\ref{eq:pout_approx}) can be rewritten as
\begin{align}
&P_{\text{out}}^{\text{approx}} = 1 - \sum_{i=0}^{k_X - 1} \frac{\sqrt{\pi}}{i!} \frac{\Gamma(2k_Y + i)}{2^{2k_Y + i - 1} \Gamma(k_Y)} \left( \frac{\gamma_{\text{th}} \theta_Y}{\bar{\gamma} \theta_X^2} \right)^{\frac{i}{2}} \notag\\
&\left( \frac{1}{\Gamma\left(k_Y + \frac{i + 1}{2}\right)} \Phi\left(-\left(k_Y + \frac{i}{2}\right), \frac{1}{2}; \frac{\gamma_{\text{th}} \theta_Y}{4 \bar{\gamma} \theta_X^2}\right)\right. \notag\\
&\left.- \frac{\sqrt{\frac{\gamma_{\text{th}} \theta_Y}{\bar{\gamma} \theta_X^2}}}{\Gamma\left(k_Y + \frac{i}{2}\right)} \Phi\left(-\left(k_Y + \frac{i + 1}{2}\right), \frac{3}{2}; \frac{\gamma_{\text{th}} \theta_Y}{4 \bar{\gamma} \theta_X^2}\right) \right).
\label{eq:pout_approx_asymptotic}
\end{align}

An approximation of $\Phi(a, b; x)$, as $x \rightarrow 0$, is given by \cite[eq.~ (07.20.03.0001.01)]{ref43},
\begin{equation}
\lim_{x \rightarrow 0} \Phi(a, b; x) = 1.
\label{eq:Phi_approx}
\end{equation}

By substituting (\ref{eq:Phi_approx}) into (\ref{eq:pout_approx_asymptotic}) and ignoring the series sequence with smaller terms due to their relatively small value, (\ref{eq:pout_approx_asymptotic}) can be rewritten as
\begin{equation}
P_{\text{out}}^{\text{approx}} = \frac{\sqrt{\pi}}{\Gamma(k_X + 1)} \frac{\Gamma(2k_Y + k_X - 1)}{2^{2k_Y + k_X - 2} \Gamma(k_Y)} \left( \frac{\left( \frac{\gamma_{\text{th}} \theta_Y}{\bar{\gamma} \theta_X^2} \right)^{\frac{k_X}{2}}}{\Gamma\left(k_Y + \frac{k_X - 1}{2}\right)} \right).
\label{eq:pout_approx_final}
\end{equation}

Then, by utilizing $\frac{\sqrt{\pi} \Gamma(n)}{2^{n - 1} \Gamma(n/2)} = \Gamma\left(\frac{n + 1}{2}\right)$, we derive (\ref{eq:pout_asymp}), thereby completing the proof.
}

\end{document}